\documentclass[twocolumn,10pt,final,journal]{IEEEtran}
\usepackage{mathrsfs}
\usepackage{cite}
\usepackage{amsfonts}
\usepackage[cmex10]{amsmath}
\usepackage{array}
\usepackage{mdwmath}
\usepackage{eqparbox}
\usepackage{graphics,dblfloatfix}
\usepackage{graphicx}
\usepackage{epsfig}
\usepackage{subfigure}
\usepackage{amssymb}
\usepackage{lineno}
\usepackage{cite}
\usepackage{booktabs}
\usepackage{tabularx}
\usepackage{amsmath}
\usepackage{multirow}
\usepackage{epsfig}
\usepackage{amssymb}
\usepackage{pgfplots}
\usepackage{algorithmic}
\usepackage[ruled]{algorithm}
\usepackage{color}
\usepackage{times}
\usepackage{paralist}
\usepackage{xspace}
\usepackage{mathrsfs}
\usepackage{clrscode}
\allowdisplaybreaks 
\newtheorem{thm}{Theorem}
\newtheorem{lem}{Lemma}
\newtheorem{cor}{Corollary}
\newtheorem{defn}{Definition}

\renewcommand{\paragraph}[1]{\smallskip \noindent {\textsc{#1}}}

\def\ie{\textit{i.e.}\xspace}
\def\etal{\textit{et al.}\xspace}

\def\whp{{\emph{w.h.p.}}}

\renewcommand{\baselinestretch}{1.08}

\newcommand{\mycut}[1]{{}}

\DeclareMathOperator{\EMST}{EMST} \DeclareMathOperator{\EST}{EST}

\def\ie{\textit{i.e.}\xspace}
\def\etal{\textit{et al.}\xspace}

\def\whp{{\emph{w.h.p.}}}

\begin{document}
\title{General Capacity for Deterministic Dissemination in Wireless Ad Hoc Networks}

\author{Cheng~Wang, Jieren Zhou, Tianci Liu, Lu Shao, Huiya Yan, Xiang-Yang Li, Changjun Jiang
 \IEEEcompsocitemizethanks{\IEEEcompsocthanksitem  Wang, Zhou, Shao and Jiang are with the Department of Computer
Science and Engineering, Tongji University, and with the Key
Laboratory of Embedded System and Service Computing, Ministry of
Education, China. (E-mail: 3chengwang@gmail.com)
\IEEEcompsocthanksitem  Liu and Li are with The School of software, Tsinghua University. (E-mail: xiangyang.li@gmail.com)
\IEEEcompsocthanksitem  Yan is with Mathematics Department, University of Wisconsin-La Crosse.
(E-mail: hyan@uwlax.edu)
\IEEEcompsocthanksitem Part of the results was published in IEEE
INFOCOM 2011 \cite{2011-cheng-general-capacity}.
}
}
\maketitle

\begin{abstract}
In this paper, we study capacity scaling laws of the deterministic dissemination (DD)
in random wireless networks under the generalized physical
model (GphyM). This is truly not a new topic.
Our motivation to readdress this issue is two-fold: Firstly, we aim to propose a more general result to unify the network capacity for general homogeneous random models by investigating the impacts of different parameters of the system on the network capacity.
Secondly, we target to close the open gaps between the upper and the lower bounds on the network capacity in the literature.
The generality of this work lies in three aspects:
(1) We study the homogeneous random network of a general node density
$\lambda \in [1,n]$, rather than either random dense
network (RDN, $\lambda=n$) or random extended network (REN,
$\lambda=1$) as in the literature. (2) We address
the general deterministic  dissemination sessions, \ie, the general multicast sessions, which unify the capacities for unicast and broadcast sessions
by setting the number of destinations for each session as a general
value $n_d\in[1,n]$. (3) We allow the number of sessions to change in
the range $n_s\in(1,n]$, instead of assuming that $n_s=\Theta(n)$ as
in the literature. We derive the general upper bounds on the capacity for the arbitrary case of $(\lambda, n_d, n_s)$ by introducing the Poisson Boolean model of continuum percolation,
and prove that they are tight according to the existing general lower bounds constructed in the literature.
\end{abstract}


\section{Introduction}\label{section:Introduction}

This work falls within the scope of the issue of capacity scaling laws for wireless networks, initiated by Gupta and Kumar~\cite{02Gupta2000}, \ie, the
scaling of network performance in
the limit when the network gets large, \cite{ozgur2007hca}.
The main advantage of studying scaling laws is to highlight
qualitative and architectural properties of the system without
considering too many details
\cite{02Gupta2000,ozgur2007hca}.
The network capacity  depends directly on the type of traffic sessions of interest.
Generally, the traffic sessions in wireless networks
can be classified into two broad types: \emph{data dissemination}, where a session has only one source, and \emph{data gathering}, where a session intends to transmit data from its multiple sources to a relatively small number of destinations;
on the other hand, according to the property of destination selection schemes, they can also be divided into the following two types:
 \emph{deterministic session}, where the selection of destination(s) are/is determined beforehand, and \emph{opportunistic session}, where the destination(s) are/is opportunistically chosen during the transmitting procedure.
Based on those classifications, the typical session patterns can be located as shown in Table. \ref{tab:session}.

In this work,  we focus on dissemination sessions that
can be usually represented by a triple dimensional vector $(n,n_c,n_d)$ with $1\leq
n_d\leq n_c \leq n-1$.
These parameters are defined by the following:
The node set of network, say $\mathcal{V}:=\mathcal{V}(n)$, comprises $n$ nodes;
the cardinality of source set
$\mathcal{S}\subseteq \mathcal{V}$ is $|\mathcal{S}|=n_s$; during
the process of dissemination with source $v_i
\in \mathcal{S}$, $n_c$ nodes are randomly chosen to compose a
\emph{candidate set}, denoted by $\mathcal{C}_i$, and the session is
completed when data are transmitted to a subset $\mathcal{D}_i\subseteq \mathcal{C}_i$, called \emph{destination set},  where
$|\mathcal{D}_i|=n_d\leq n_c$.
Obviously, when $n_d\equiv n_c$, the
 dissemination  is specified into a \emph{deterministic dissemination}, \ie,
 the so-called \emph{general multicast session}.
Please see the illustration in
Fig.\ref{fig-manycast}.

The purpose of this paper is to investigate the capacity of wireless networks
where $n_s:(1,n]$\footnote{We use the term $f(n): [\phi_1(n), \phi_2(n)]$
to represent $f(n)=\Omega(\phi_1(n))$ and $f(n)=O(\phi_2(n))$; and
use $f(n): (\phi_1(n), \phi_2(n))$ to represent
$f(n)=\omega(\phi_1(n))$ and $f(n)=o(\phi_2(n))$.} general multicast sessions, denoted by $(n,n_d,n_d)$ with $n_d$:$[1,n]$, run simultaneously.
In the research of \emph{networking-theoretic} capacity scaling laws
\cite{ozgur2007hca}, the unicast and broadcast sessions can be usually regarded as two special cases of general multicast sessions according to
the number of destinations for each session.
Usually, any proposed multicast capacity could be specialized into the unicast and
broadcast capacities by letting $n_d=1$ and $n_d=n$, respectively.
This principle often applies in the literatures
\cite{01Xiang-Yang2007,Alireza-2008,LiMobiCom08,Hu-hoc09,2012TC-MC,Multicast2014}.

\begin{figure}[t]
\begin{center}
\scalebox{1}{\includegraphics{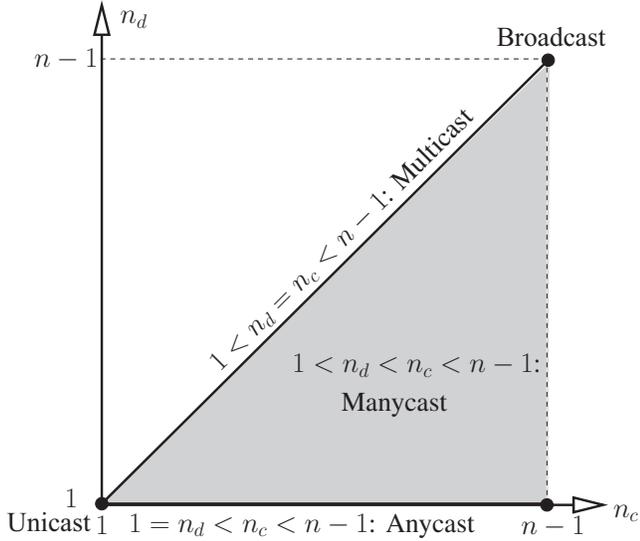}}  \caption{General
Dissemination Sessions. Here, $m$ is the size of
the specified group for a source, where $1\leq m\leq n-1$; $d$ is the number
of desired destinations of the source, where $1\leq d\leq m$.} \label{fig-manycast}
\end{center}
\end{figure}

Most of the existing results differ from each other due to the diversity of adopted analytical models and assumptions. Besides session patterns introduced above,
there are two typical models in terms of scaling patterns that are adopted in the literature: \emph{random extended network} (REN),
where the node density is fixed to a constant \cite{15Franceschetti2007,LiMobiCom08,Zhengrong2008transaction,2012TC-MC},
 and \emph{random dense network} (RDN),
 where  the node density increases linearly with the number of nodes \cite{02Gupta2000,16Keshavarz-Haddad2006,20Shakkottai2007,Alireza-arena2007,Alireza-2008}.
In \cite{20Shakkottai2007}, Shakkottai \etal derived the multicast capacity of RDN for a specific case
that $n_s=n^\epsilon$ and $n_s\cdot n_d=\Theta(n)$, where
$\epsilon\in(0,1]$. They showed that such per-session multicast capacity under
the protocol model is at most of $O(\frac{1}{\sqrt{n_s\log n}})$. To achieve the upper bound, they proposed a simple and novel
routing architecture, called the \emph{multicast comb}, to transfer
multicast data in the network. A more general result, in terms of $n_s$ and $n_d$, was proposed by Li \etal in
\cite{01Xiang-Yang2007}. They showed that when $n_s=\Omega(\log
n_d\cdot \sqrt{n\log n/n_d})$, the per-session multicast capacity
for RDN under the protocol model is of
$\Theta(\frac{1}{n_s}{\sqrt{\frac{n}{n_d\log n }}})$ if
$n_d=O(\frac{n}{\log n})$, and is of  $\Theta( {1}/{n_s})$ if
$n_d=\Omega({n}/{\log n})$. After that, Keshavarz-Haddad \etal
\cite{Alireza-2008} derived the multicast capacity for RDN under
the \emph{generalized physical model} \cite{agarwal:cba} by
designing new multicast schemes and computing the upper bounds. A \emph{gap} remains open between the
upper and the lower bounds in the regime $n_d:[n/(\log n)^3, n/\log n]$ as illustrated in Fig.\ref{fig-existingresult}(a).
For multicast capacity of REN under the generalized physical model,
Li \etal \cite{LiMobiCom08} derived a lower bound as
$\Omega(\frac{\sqrt{n}}{n_s\sqrt{n_d}})$ for the case that
$n_s=\Omega(n^{1/2+\epsilon})$ and $n_d=O(n/(\log n)^{2\alpha+6})$.
Recently, Wang \etal \cite{2012TC-MC} devised the specific
multicast schemes and derived the multicast throughput for all cases
$n_s$:$(1,n]$ and $n_d$:$[1,n]$. Under the assumption that
$n_s=\Theta(n)$, their lower bounds are specialized into those in
Equation (\ref{equ-REN-throughput}). They also derived an upper
bound for the case that $n_s=\Theta(n)$, as in Equation
(\ref{equ-REN-upperbound}). An obvious \emph{gap} exists between the
upper and the lower bounds in the regime $n_d$:$[n/(\log n)^{\alpha+1},
n/\log n]$ (Please see the illustration in
Fig.\ref{fig-existingresult}(b)).
To the best of our knowledge, this is the latest results on general multicast capacity for static REN without considering the impacts of node mobility \cite{zhou2010delay,Hu-hoc09} or advanced physical communication technology \cite{MASS-2009-Cognitive,wangxinbing1}.
\emph{Closing} the remaining \emph{gaps}
is one of the motivations of this paper.

Both REN and RDN are extreme cases for a random network consisting of $n$ nodes in terms of the node density $\lambda$. So the
characterization of two particular models does not suffice to
develop a comprehensive understanding of wireless networks, although
they are representative models to some extent, \cite{ozgur2007hca}.
Hence, in this paper, we comprehensively consider the network with a
general node density $\lambda:[1,n]$, rather than only the cases
$\lambda=1$ (REN) and $\lambda=n$ (RDN), which can offer complete
and deep insights about the scaling laws for wireless networks.
\emph{Unearthing} the \emph{nature} of general scaling is another
motivation of this work.

In conclusion, we aim to examine the capacity scaling laws
of  general wireless networks, where the generality lies in
three aspects: (1) a general
node density, $\lambda$:$[1,n]$; (2) a general number of receivers,
$n_d:[1,n]$; (3) a general number of sessions, $n_s$:$(1,n]$.
For such general multicast capacity of general wireless networks,
we have computed the lower bounds under the generalized physical model in \cite{2011-cheng-general-capacity}.
More specifically, we build routing backbones of two levels: \emph{highways} and
\emph{arterial roads}. Furthermore,  arterial roads (ARs) have two
subclasses, \ie, \emph{ordinary arterial roads} (O-ARs) and
\emph{parallel arterial roads} (P-ARs). Note that the highways are
the same as those in \cite{15Franceschetti2007,LiMobiCom08,Alireza-2008,{2012TC-MC}},
but the ARs are different from the \emph{second-class highways}
(SHs) in \cite{2012TC-MC}. Recall that in the SH system of
\cite{2012TC-MC}, there are two types of  SHs: \emph{odd}
SHs and \emph{even} SHs. The bottleneck of the whole routing could
happen in the switching phase between the  odd  and  even  SHs.
There is no such a bottleneck in the current AR system, which can
improve the multicast throughput for some regimes of $n_s$ and
$n_d$. Based on the highways, O-ARs and P-ARs, we design four
routing schemes. By exploiting the theory of \emph{maximum
occupancy}, we derive the optimal multicast throughput and scheme
according to different ranges of $\lambda$, $n_d$, and $n_s$.

\begin{table} \renewcommand{\arraystretch}{1.31}
\caption{Typical Session Patterns} \label{tab:session}
\vspace{-0.1in}
 \centering
 \scalebox{0.96}{\begin{tabular}{c|c|c}
 \hline   & Deterministic Session & Opportunistic Session\\
  \hline \hline
\multirow{3}{0.8in}{Dissemination /Single-Source} & Unicast           &                         Anycast            \\
& Broadcast           &                 $\cdots$                     \\
& Multicast           &                 Manycast                    \\
\hline
\multirow{3}{0.8in}{Gathering /Multiple-Sources}&     Data Collection       &         Undefined                            \\
&    $\cdots$ &          (to the best of  \\
& ConvergeCast  (Many-to-One) &                       ~~~~our knowledge)              \\
  \hline \end{tabular}
  }
  \end{table}

\textbf{Major contributions of this paper} can be summarized as follows:

  $\triangleright$  For deriving the upper bounds on multicast capacity, we introduce the Poisson Boolean model of continuum percolation \cite{meester1996cp} (not
Poisson bond percolation model \cite{15Franceschetti2007}), which, to the best of our knowledge, is not used in previous studies on upper bounds of network capacity.
Based on the argument of \emph{giant cluster} (component) in the Poisson boolean percolation model,
we can divide the communications under any multicast routing scheme into two parts, \ie,
communications inside  and outside the giant component.
Obviously, the network throughput must be determined by the bottleneck of two parts.
We give a general formula to compute upper bounds on the capacity.

  $\triangleright$ For the case that $n_s=\Theta(n)$ and $\lambda=n$ (or $\lambda=1$),
\ie, RDN and REN, due to the limitations of adopted analytical methods,
the previous works \cite{Alireza-2008,{2012TC-MC}} have not derived the tight bounds
on  multicast capacity under the generalized physical model.
By applying our general results to these special cases,
we close those gaps.

The rest of the paper is organized as follows. The system model is
formulated  in Section \ref{section:SystemNetwork}. We present and
discuss the main results in Section \ref{sec-main-results}. In
Section \ref{sec-technical-preparations}, we make preparations for
the analysis. We derive the upper bounds on the capacity
in Section \ref{sec-upper}.
For completeness, we include the derivation of lower bounds from \cite{2011-cheng-general-capacity} in Appendix \ref{sec-lower}.
We draw some conclusions  in Section
\ref{section-Conclusion and Future Work}.

\section{System Model}\label{section:SystemNetwork}

\subsection{Random Scaling Model}

We construct a random network
with node density $\lambda$, denoted by $\mathcal{N}(\lambda,n)$, by placing wireless nodes randomly into a square region
$\mathcal{R}(\lambda,n)=[0,\sqrt{A}]^2$ according to a Poisson point process with density $\lambda$, where $A= n/\lambda $. When
$\lambda$ is set to be $1$ (or $n$), our model corresponds to
\emph{random extended network} (REN) (or \emph{random dense network}
(RDN)).
According to Chebyshev's
inequality, we get that the number of nodes in $\mathcal{A}(a^2)$ is
within $((1-\epsilon )n, (1+\epsilon)n)$ with high probability,
where $\epsilon>0$ is an arbitrarily small constant. To simplify the
description, we assume that the number of nodes is exactly $n$,
without changing our results in the sense of order,
\cite{15Franceschetti2007,Zhengrong2008transaction,2012TC-MC}. We are mainly
concerned with the events that occur inside these squares with high
probability (w.h.p.); that is, with probability approaching one as $n
\to \infty$.

\begin{figure*}[t]
\begin{center}
\begin{tabular}{cc}
\scalebox{0.9}{\includegraphics{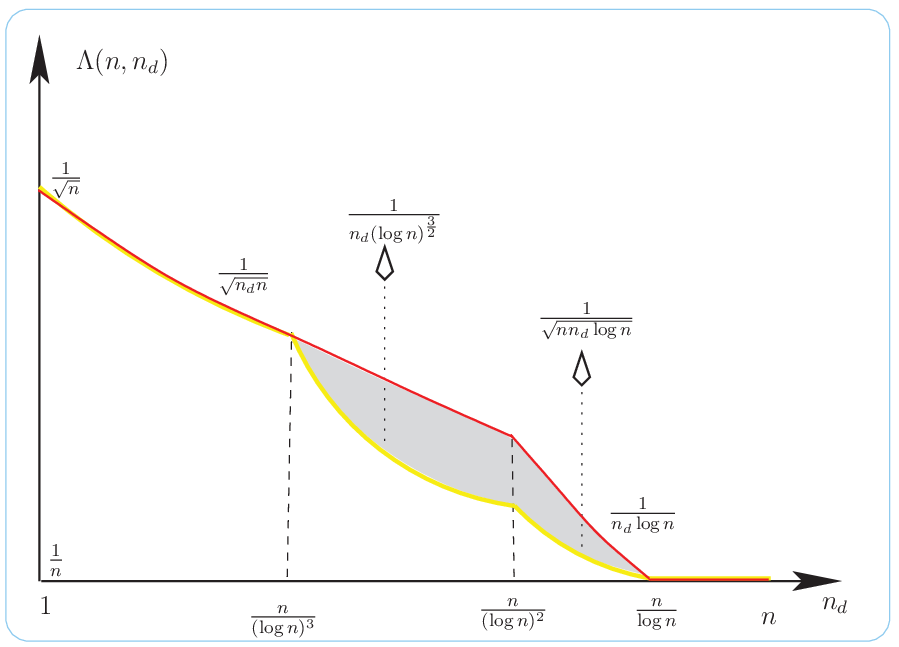}}    & \scalebox{0.9}{\includegraphics{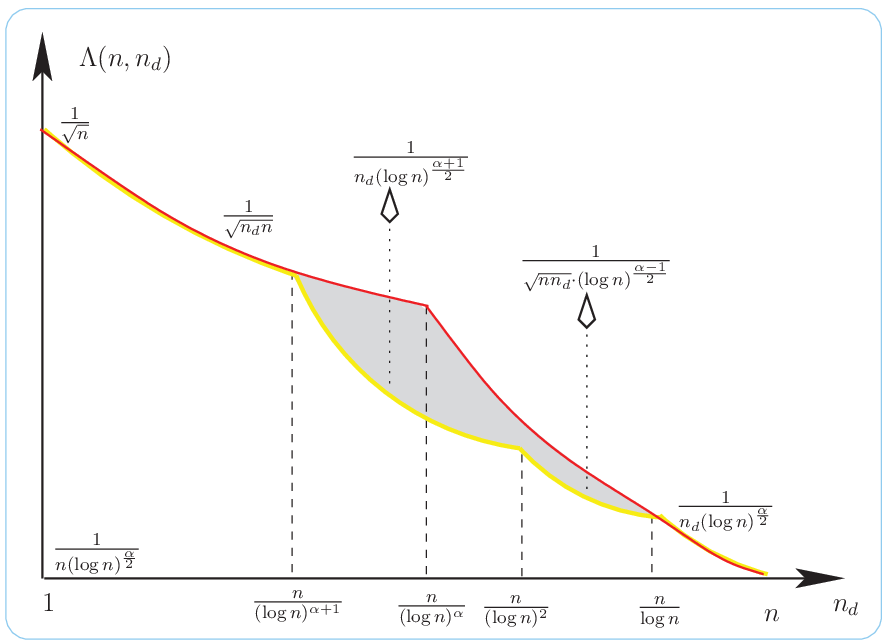}}\\
{\small (a) Capacity for RDN }&
{\small (b) Capacity for REN }
\end{tabular}
\end{center}
\caption{Results on General Multicast Capacity.
Obvious \emph{gaps} exist between the upper and the lower
bounds in the regimes $n_d  :[n/(\log n)^3,
n/\log n]$ for RDN and $n_d  :[n/(\log n)^{\alpha+1}, n/\log n]$
for REN, illustrated by the shaded regions.}
\label{fig-existingresult}
 \end{figure*}

\subsection{Session Patterns}\label{subsec-sessionpattern}

In wireless networks, there are two broad types of session patterns:
\emph{information dissemination} and \emph{information gathering}.
The former is the interest of this paper. Generally, dissemination
sessions can be further divided into two categories:
\emph{deterministic dissemination}, in which the destination(s) of a
message is (are) determined when it is generated at a source, such
as unicast, broadcast, and multicast, and \emph{opportunistic
dissemination}, such as anycast
\cite{xuan2002routing,choudhury2004mac}, and manycast
\cite{carter2003manycast} sessions, in which the destination(s) of a
message is (are) opportunistically chosen and both the paths to  the
group member(s) and the destination(s) can change dynamically
according to the network condition, such as the node movement
situation.

In this work, we focus on the general multicast sessions,
including unicast, broadcast and multicast sessions. We adopt a similar construction procedure to the one in
\cite{2012TC-MC}.
To generate the $k$-th ($1 \leq k\leq n_s$) multicast session, with source $v_{\mathcal{S},k}\in \mathcal{S}$,
denoted by $\mathcal{M}_{\mathcal{S},k}$, $n_d$ points
$p_{\mathcal{S},k_i}$ ($1 \leq i \leq n_d$, and $1 \leq n_d \leq
n-1$) are randomly and independently chosen from the deployment
region $\mathcal{R}(\lambda,n)$. Denote the set of these $n_d$
points by $\mathcal{\tilde{P}}_{\mathcal{S},k}=\{p_{\mathcal{S},k_1},
 p_{\mathcal{S},k_2},\cdots,   p_{\mathcal{S},k_{n_d}}\}$.
Let $v_{\mathcal{S},k_i}$ be
the nearest ad hoc node from $p_{\mathcal{S},k_i}$ (ties are broken randomly). In
$\mathcal{M}_{\mathcal{S},k}$,  the node $v_{\mathcal{S},k}$, serving as a source, intends to deliver data to $n_d$ destinations
$\mathcal{D}_{\mathcal{S},k}=\{v_{\mathcal{S},k_1},
 v_{\mathcal{S},k_2},\cdots,   v_{\mathcal{S},k_{n_d}}\}$  at an arbitrary data rate $\lambda_{{\mathcal{S}},k}$.
Let $\mathcal{U}_{\mathcal{S},k}=\{v_{\mathcal{S},k}\}\cup
\mathcal{D}_{\mathcal{S},k}$
 be the \emph{spanning set} of nodes for the multicast session $\mathcal{M}_{\mathcal{S},k}$.
Please see the illustration in Fig.\ref{fig-construct-m}.

\begin{figure}[t]
\begin{center}
\scalebox{1}{\includegraphics{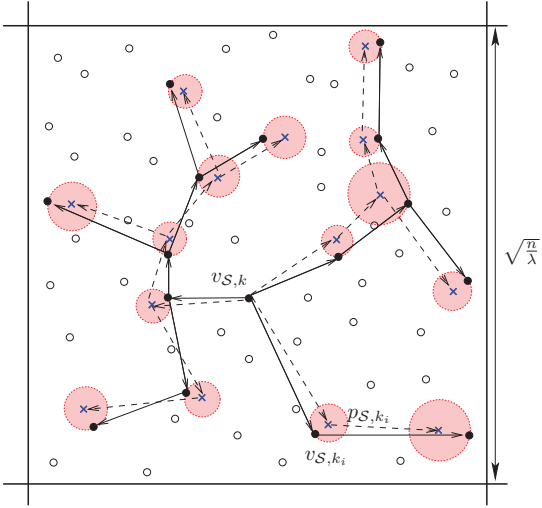}}
\caption{Multicast
session $\mathcal{M}_{\mathcal{S},k}$, \cite{2012TC-MC}. The tree consisting of solid
lines represents the Euclidean minimum spanning tree (EMST) over
$\mathcal{U}_{\mathcal{S},k}=\{v_{\mathcal{S},k_i}~|~ 0\leq i \leq
n_d\}$, denoted by $\EMST(\mathcal{U}_{\mathcal{S},k})$, where $v_{\mathcal{S},k_0}$ is $v_{\mathcal{S},k}$. The tree
consisting of dashed lines represents an Euclidean spanning tree
(EST) over $\mathcal{P}_{\mathcal{S},k}=\{p_{\mathcal{S},k_i}~|~
0\leq i \leq n_d\}$, denoted by
$\EST_0(\mathcal{P}_{\mathcal{S},k})$, where $p_{\mathcal{S},k_0}$ is $v_{\mathcal{S},k}$, and for any $0 \leq i,j
\leq n_d$, link $(p_{\mathcal{S},k_i} \to p_{\mathcal{S},k_j}) \in
\EST_0(\mathcal{P}_{\mathcal{S},k})$ if and only if
link $(v_{\mathcal{S},k_i} \to v_{\mathcal{S},k_j}) \in
\EMST(\mathcal{U}_{\mathcal{S},k})$.
} \label{fig-construct-m}
\end{center}
\end{figure}

\subsection{Communication Model}\label{subsec-interferencemodel}
Generally, there are three types of communication (interference)
models: the \emph{protocol model}\cite{02Gupta2000},
 \emph{physical model} \cite{02Gupta2000}
and \emph{generalized physical model} \cite{agarwal:cba} (along with the name
¡°Gaussian Channel model¡±, \cite{LiMobiCom08}). We adopt
the generalized physical model since it is more realistic than the
other two
\cite{agarwal:cba,15Franceschetti2007,LiMobiCom08,Alireza-arena2007}.

Let $\mathcal{K}_t$ denote a \emph{scheduling set} of links in which all
links can be scheduled simultaneously in time slot $t$.
\begin{defn}\label{defn-Gaussianmodel}
Under the generalized physical model, when a scheduling set
$\mathcal{K}_t$ is scheduled,
  the rate of a link $<u,v>\in \mathcal{K}_t$ is achieved at
\begin{equation}\label{equ-gaussianchannelmodel}
R_{u,v;t}=B\times\mathbf{1}\cdot\{<u,v>\in \mathcal{K}_t\} \times
\log(1+\mathrm{SINR}_{u,v;t}),
\end{equation}
where $\mathrm{SINR}_{u,v;t}=\frac{P\cdot
\ell(|\mathbf{x}_u-\mathbf{x}_v|) }{N_0+\sum_{<i,j>\in
\mathcal{K}_t{}/<u,v>}
  P\cdot \ell(\mathbf{x}_i-\mathbf{x}_v|)}$; $\mathbf{x}_u$ denotes the position of node $u$,
  $|\mathbf{x}_u-\mathbf{x}_v|$ represents the Euclidean distance between node $u$ and node $v$;
   $\ell(\cdot)$ denotes the power attenuation function that is assumed to depend only on the distance between the
    transmitter and the receiver
\cite{02Gupta2000,15Franceschetti2007,LiMobiCom08,chau2009capacity};
 $\ell(|\cdot|):=|\cdot|^{-\alpha}$ for dense scaling networks, and
$\ell(|\cdot|):=\min\{1, |\cdot|^{-\alpha}\}$ for extended scaling
networks \cite{15Franceschetti2007}.
\end{defn}

\section{Main Results}\label{sec-main-results}
We mainly derive the upper bounds on the general multicast capacity of random ad hoc networks.

\subsection{General Upper Bounds}

\begin{thm}\label{thm-upper-bound}
The multicast capacity for random network $\mathcal{N}(\lambda,n)$
is at most
\[ \overline{\Lambda} (\lambda,n)=\max\limits_{l_c:\mathcal{L}_c} \left\{\min\left\{\frac{\min\{1, l_c^{-\alpha}\}}{\mathbf{L}(n_s,  \frac{\sqrt{n}}{l_c\sqrt{n_d \lambda}})},
\frac{\min\{1, (\frac{\lambda}{\log
n})^{\frac{\alpha}{2}}\}}{\mathbf{L}(n_s,  \frac{ n \cdot\sqrt{ \lambda}
\cdot l_c}{n_d \cdot \sqrt{\log n}})}\right\}\right\},\]
where $\mathcal{L}_c=[{1}/{\sqrt{\lambda}},\sqrt{{\log n}/{\lambda} }]$,  and $\mathbf{L}(m,n)$ is defined in Table \ref{tab-functions}.
\end{thm}

\subsection{Tight Capacity Bounds} 


In \cite{2011-cheng-general-capacity}, the general lower bounds have been provided by designing some strategies.

\begin{lem}[\cite{2011-cheng-general-capacity}]\label{thm-achievable-MT-all-schemes}
The general multicast throughput for random network
$\mathcal{N}(\lambda,n)$ can be achieved as
\[\underline{\Lambda}(\lambda,n)=\max\{\Lambda_{\mathrm{o}}(\lambda, n),\Lambda_{\mathrm{p}}(\lambda, n), \Lambda_{{\mathrm{o}\&\mathrm{h}}}(\lambda, n),\Lambda_{{\mathrm{p}\&\mathrm{h}}}(\lambda, n)\},\]
where $\Lambda_{\mathrm{o}}(\lambda,
n),\Lambda_{\mathrm{p}}(\lambda, n),
\Lambda_{{\mathrm{o}\&\mathrm{h}}}(\lambda,
n),\Lambda_{{\mathrm{p}\&\mathrm{h}}}(\lambda, n)$ are defined in
Table \ref{tab-functions}.
\end{lem}

\begin{table}[t] \renewcommand{\arraystretch}{1.8}
\caption{Defined Functions and Parameters.} \label{tab-functions}
 \centering
 \scalebox{1}{\begin{tabular}{p{1.4cm}p{6.4cm}}
 \hline Functions  &Definitions\\
  \hline \hline
      $\mathbf{L}(m, n)$  & $\left\{ \begin{array}{lrl}
 \Theta\left(\frac{\log n}{\log \frac{n}{m}}\right) &
 \mbox{when}& m:[1, \frac{n}{\mathrm{polylog}(n)})\\
  \Theta\left(\frac{\log n}{\log \frac{n\log n}{m}}\right) &
 \mbox{when}& m:[\frac{n}{\mathrm{polylog}(n)}, n\log n)\\
 \Theta\left(\frac{m}{n}\right) & \mbox{when}& m =\Omega(n\log n)\\
      \end{array} \right.  $                \\
                   \hline
 $ \mathbf{R}_{\mathrm{O-AR}}(\lambda, n)$    & $ \left\{
\begin{array}{lll}
\Theta(\frac{\lambda^{\frac{\alpha}{2}}}{(\log n)^{\frac{\alpha}{2}}}) & \mbox{when} &  \lambda :[1, \log n] \\
\Theta(1) & \mbox{when} &   \lambda :[\log n, n] \\
  \end{array} \right.
$               \\
                                     \hline
$ \mathbf{R}_{\mathrm{P-AR}}(\lambda, n)$   & $\left\{
\begin{array}{lll}
\Theta(\frac{\lambda^{\frac{\alpha}{2}}}{(\log n)^{\frac{\alpha}{2}}}) & \mbox{when} &  \lambda :[1, (\log n)^{1-\frac{2}{\alpha}}] \\
\Theta(\frac{1}{\log n} ) & \mbox{when} &   \lambda :[(\log n)^{1-\frac{2}{\alpha}}, n] \\
  \end{array} \right.
$                 \\
                                     \hline
$\mathbf{p}_{\mathrm{o}}$  & $\left\{ \begin{array}{ll}
\Theta(\sqrt{ \frac{n_d \log n}{n} }) & \mathrm{when}~n_d   =   O(\frac{n}{\log n}) \\
\Theta(1) & \mathrm{when}~n_d=\Omega(\frac{n}{\log n})\\
  \end{array} \right.  $               \\
                   \hline
   $\mathbf{p}_{\mathrm{p}}$   &   $\left\{ \begin{array}{ll}
\Theta(\frac{\sqrt{n_d}}{\sqrt{n\log n}}) & \mathrm{when}~n_d :[1,\frac{n}{\log n}]   \\
\Theta(\frac{n_d}{n}) & \mathrm{when}~n_d :[\frac{n}{\log n},n]\\
  \end{array} \right.  $               \\
                   \hline
  $\mathbf{p}_{\mathrm{oh},\mathrm{O-AR}}$   &  $\left\{ \begin{array}{ll}
\Theta( \frac{n_d \cdot (\log n)^{3/2}}{n} ) & \mathrm{when}~n_d :[1,\frac{n}{(\log n)^{3/2}}]   \\
\Theta(1) & \mathrm{when}~n_d :[\frac{n}{(\log n)^{3/2}},n]\\
  \end{array} \right.  $                  \\
\hline
 $\mathbf{p}_{\mathrm{oh},\mathrm{H}}$,  $\mathbf{p}_{\mathrm{ph},\mathrm{H}}$  & $\left\{ \begin{array}{ll}
\Theta(\sqrt{\frac{n_d}{n}}) & \mathrm{when}~n_d:[1,\frac{n}{(\log n)^2}] \\
\Theta(\frac{n_d\log n}{n}) & \mathrm{when}~n_d:[\frac{n}{(\log n)^2},\frac{n}{\log n}]\\
\Theta(1) & \mathrm{when}~n_d:[\frac{n}{(\log n)},n]\\
  \end{array} \right.     $             \\
\hline $\mathbf{p}_{\mathrm{ph},\mathrm{P-AR}}$   & $\left\{
\begin{array}{ll}
\Theta( \frac{n_d \cdot  \sqrt{\log n} }{n} ) & \mathrm{when}~n_d :[1,\frac{n}{ \sqrt{\log n} }]   \\
\Theta(1) & \mathrm{when}~n_d :[\frac{n}{ \sqrt{\log n} },n]\\
  \end{array} \right.  $ \\
\hline
$\Lambda_{\mathrm{o}}(\lambda, n)$  & ~${\mathbf{R}_{\mathrm{O-AR}}(\lambda, n)}/{\mathbf{L}(n_s,\frac{1}{\mathbf{p}_\mathrm{o}})}$\\
\hline
$\Lambda_{\mathrm{p}}(\lambda, n)$ & ~${\mathbf{R}_{\mathrm{P-AR}}(\lambda, n)}/{\mathbf{L}(n_s,\frac{1}{\mathbf{p}_\mathrm{p}})}$\\
\hline $\Lambda_{{\mathrm{o}\&\mathrm{h}}}(\lambda, n)$ &
~$\min\left\{\frac{\mathbf{R}_{\mathrm{O-AR}}(\lambda,
n)}{\mathbf{L}(n_s,
\frac{1}{\mathbf{p}_{\mathrm{oh},\mathrm{O-AR}}})},
\frac{1}{\mathbf{L}(n_s,\frac{1}{\mathbf{p}_{\mathrm{oh},\mathrm{H}}})}\right\}$ \\
\hline $\Lambda_{{\mathrm{p}\&\mathrm{h}}}(\lambda, n)$ &
~$\min\left\{\frac{\mathbf{R}_{\mathrm{P-AR}}(\lambda,
n)}{\mathbf{L}(n_s,
\frac{1}{\mathbf{p}_{\mathrm{ph},\mathrm{P-AR}}})},
\frac{1}{\mathbf{L}(n_s,\frac{1}{\mathbf{p}_{\mathrm{ph},\mathrm{H}}})}\right\}$\\
\hline
  \end{tabular}
  }
  \end{table}

We specialize the general results from Theorem \ref{thm-upper-bound}
and Lemma \ref{thm-achievable-MT-all-schemes} to the cases that $\lambda=n$ and  $\lambda=1$, corresponding to the RDN and REN.
Following a common  assumption in most existing works, \ie,
$n_s=\Theta(n)$, we  show that for both RDN and REN
our results give the first tight bounds on multicast capacity over the whole regime $n_d:[1,n]$.

\subsubsection{Random Dense Networks}
In Theorem \ref{thm-upper-bound},
$\overline{\Lambda} (n,n)$, \ie,
the upper bound on the capacity achieves its maximum value
by choosing $l_c=\Theta(\frac{1}{\sqrt{n}})$ when $n_d=O( {n}/{(\log n)^2})$; and also achieves its maximum value by choosing $l_c=\Theta(\sqrt{\log n}/\sqrt{n})$
when $n_d=\Omega( {n}/{(\log n)^2})$.
Specifically, the multicast capacity  is
  \emph{at most} of order
  \begin{equation}\label{equ-ns=n-throughput-2}
\left\{ \begin{array}{lrl}
\Theta(\frac{1}{\sqrt {n_d
n}}) & \mathrm{when}& n_d:[1,\frac{n}{({\log n})^3}]\\
 \Theta(\frac{1}{n_d(\log n)^{\frac{3}{2}}}) &
 \mathrm{when}& n_d:[\frac{n}{({\log n})^3},\frac{n}{(\log n)^2}]\\
 \Theta(\frac{1}{\sqrt{n n_d \log n}}) & \mathrm{when} & n_d :[\frac{n}{(\log n)^2},\frac{n}{\log n}]\\
\Theta(\frac{1}{n}) & \mathrm{when}& n_d :[\frac{n}{\log n},n]\\
   \end{array} \right.
\end{equation}

This result is \emph{exciting},  because the multicast
throughput as in Equation (\ref{equ-ns=n-throughput-2})  had been proven to
be achievable by Keshavarz-Haddad \etal in \cite{Alireza-2008}. Moreover, they derived an upper bound as
\begin{equation}\label{equ-RDN-A-upperbound}
\left\{ \begin{array}{lrl} O(\frac{1}{\sqrt {n_d n}}) &\mathrm{when}
& n_d :  [1,\frac{n}{({\log n})^{2}}]
\\
O(\frac{1}{n_d\cdot \log n}) & \mathrm{when} & n_d:
[\frac{n}{({\log n})^{2}},\frac{ n}{\log n} ]
\\
O(\frac{1}{n}) &\mathrm{when} & n_d :  [\frac{ n}{\log n}, n]
  \end{array} \right.
\end{equation}
It is clear that there is a gap between the upper and the lower
bounds in the regime $n_d:(\frac{n}{({\log n})^3}, \frac{n}{\log
n})$, as illustrated in Fig.\ref{fig-existingresult}(a).
In this work, we \emph{close} this gap. Moreover, by
Lemma \ref{thm-achievable-MT-all-schemes}, this optimal throughput
in Equation (\ref{equ-ns=n-throughput-2}) can also be achieved by using our schemes $\mathbb{M}_{\mathrm{o}}$ cooperatively and
$\mathbb{M}_{\mathrm{o}\&\mathrm{h}}$ that are defined in
Table \ref{tab-notions} in Appendix \ref{sec-lower}.

\subsubsection{Random Extended Networks}
In Theorem \ref{thm-upper-bound},
$\overline{\Lambda} (1,n)$  achieves its maximum value
by letting $l_c=\Theta(1)$ when $n_d=O( {n}/{(\log n)^2})$; and achieves its maximum value by letting $l_c=\Theta(\sqrt{\log n})$
when $n_d=\Omega( {n}/{(\log n)^2})$.
Specifically,  the
  multicast capacity  is
  \emph{at most} of order
  \begin{equation}\label{equ-REN-throughput}
\left\{ \begin{array}{lrl} \Theta(\frac{1}{\sqrt {n_d
n}}) & \mathrm{when}& n_d:[1,\frac{n}{({\log n})^{\alpha+1}}]\\
 \Theta(\frac{1}{n_d(\log n)^{\frac{\alpha+1}{2}}}) &
 \mathrm{when}& n_d:[\frac{n}{({\log n})^{\alpha+1}},\frac{n}{(\log n)^2}]\\
   \Theta(\frac{1}{\sqrt{nn_d}\cdot(\log n)^\frac{\alpha-1}{2}})& \mathrm{when} & n_d :[\frac{n}{(\log n)^2},\frac{n}{\log n}]\\
 \Theta(\frac{1}{n_d(\log n)^{\frac{\alpha}{2}}}) & \mathrm{when}& n_d :[\frac{n}{\log n},n]\\
   \end{array} \right.
\end{equation}

Also, such multicast throughput had been achieved by the schemes in
\cite{2012TC-MC}, and the upper bounds were proposed as:
\begin{equation}\label{equ-REN-upperbound}
\left\{ \begin{array}{lrl} O(\frac{1}{\sqrt {n_d
n}}) & \mathrm{when}& n_d:[1,\frac{n}{({\log n})^\alpha}]\\
O(\frac{1}{n_d(\log n)^{\frac{\alpha}{2}}})&
 \mathrm{when}& n_d:[\frac{n}{({\log n})^\alpha},n]\\
   \end{array} \right.
\end{equation}
As illustrated in Fig.\ref{fig-existingresult}(b), we close the gap between the upper and the lower bounds in the regime $n_d:[\frac{n}{({\log n})^{\alpha+1}},\frac{n}{\log n}]$.
In addition, by Lemma \ref{thm-achievable-MT-all-schemes}, this optimal throughput in Equation (\ref{equ-REN-throughput})
can  be equally achieved by using our schemes $\mathbb{M}_{\mathrm{p}}$  and $\mathbb{M}_{\mathrm{p}\&\mathrm{h}}$ cooperatively
that are defined in Table \ref{tab-notions} in Appendix \ref{sec-lower}.

\section{Technical Preparations}\label{sec-technical-preparations}

\subsection{Maximum Occupancy}
We use the results in \emph{maximum occupancy theory} to derive the
lower bounds of the multicast throughput. Now we introduce the following result from \cite{raab-balls},
\cite{mitzenmacher1996power} and \cite{liu2008data}.

\begin{lem}\label{lem-maximum-occupancy}
Let $\mathbf{L}(m, n)$ be the random variable that counts the
maximum number of balls in any bin, if we throw $m$ balls
independently and uniformly at random into $n$ bins. Then,
it holds that
\whp,
\begin{equation*}\label{equ:maximum-occupancy}
\mathbf{L}(m, n)=\left\{ \begin{array}{lrl}
 \Theta\left(\frac{\log n}{\log \frac{n}{m}}\right) &
 \mbox{when}& m:[1, \frac{n}{\mathrm{polylog}(n)})\\
  \Theta\left(\frac{\log n}{\log \frac{n\log n}{m}}\right) &
 \mbox{when}& m:[\frac{n}{\mathrm{polylog}(n)}, n\log n)\\
 \Theta\left(\frac{m}{n}\right) & \mbox{when}& m =\Omega(n\log n)\\
      \end{array} \right.
\end{equation*}
\end{lem}

\subsection{Network Throughput by Occupancy Theory}

We give a technical lemma as a basic argument of the analysis of
network capacity.

\begin{lem}\label{lem-basic-link-throughput}
Given a multicast scheme $\mathbb{M}$, for any link initiating from
a node $u$, say $uv$, if it can sustain a rate of
$\mathbf{R}(\lambda,n)$, and any multicast session shares the
bandwidth of link $uv$ with the probability of $\mathbf{p}$, then the
throughput along link $uv$  is of order $\Theta(\Lambda(\lambda,n))$,
where
$\Lambda(\lambda,n)=\frac{\mathbf{R}(\lambda,n)}{\mathbf{L}(n_s,
{1}/{\mathbf{p}})}$.
\end{lem}

\subsection{The Tail of Poisson Trials}

\begin{lem}[\cite{motwani1995randomized}]\label{lem-poisson-trial-tail}
Let $X_1$, $X_2$, $\cdots$, $X_n$ be independent Poisson trials such that, for $1\leq i \leq n$, $\Pr[X_i=1]=p_i$, where
$0<p_i<1$. Then, for $X=\sum_{i=1}^n X_i$, $\mu=\mathrm{E}(X)=\sum_{i=1}^n p_i$, and any $\delta>0$,
\[
\Pr[X>(1+\delta)\mu]<\left[ \frac{e^\delta}{(1+\delta)^{1+\delta}}\right]^\mu.
\]
\end{lem}

\subsection{Euclidean Spanning Tree}

\begin{lem}[\cite{2012TC-MC}]\label{lem-K-EMST-length}
If $X_i$, $1\leq i \leq \infty$, are uniformly distributed on
$[0,a]^d$, for a set $\mathcal{U}(n)=\{X_1, X_2,\cdots, X_n\}$,
denote its Euclidean minimum spanning tree (EMST) by
$\EMST(\mathcal{U}(n))$. Under such deployment model,
build $K(n)$ sets, denoted by  $\mathcal{U}_1(n),
\mathcal{U}_2(n), \cdots,\mathcal{U}_{K(n)}(n)$, it holds that
\begin{equation}\label{equ-K-EMST-length}
    \Pr\left[\lim_{n\to \infty} \frac{\sum_{k=1}^{K(n)}\|\EMST(\mathcal{U}_k(n))\|}{{K(n)}\cdot a \cdot n^{1-\frac{1}{d}}} = \nu(d)   \right]=1.
\end{equation}
\end{lem}

This lemma can be straightforwardly proven according to Theorem 2 of \cite{steele1988growth}.
Please see the detailed proof of Lemma D in the appendices of \cite{2012TC-MC}.

For any $n_s$ general multicast sessions constructed by the method in Section \ref{subsec-sessionpattern},
by a similar procedure to Lemma 7 of  \cite{2012TC-MC},
we have,
\begin{lem}\label{lem-n-s-tree-length-o}
For all multicast sessions $\mathcal{M}_{\mathcal{S},k}$ ($ 1 \leq k \leq n_s$),
it holds that for $n_d=o(\frac{n}{\log n})$,
\[\sum\nolimits_{k=1}^{n_s}  \|\EMST(\mathcal{D}_{\mathcal{S},k})\| =\Omega(n_s\cdot\sqrt{n_d \cdot n}),\]
where $\EMST(\mathcal{D}_{\mathcal{S},k})$ denotes the Euclidean minimum spanning tree (EMST) over the destination set
$\mathcal{D}_{\mathcal{S},k}$.
\end{lem}

Note that the session construction in this work is different from that in  \cite{2012TC-MC}, and Lemma \ref{lem-n-s-tree-length-o}
only gives a result on $\sum\nolimits_{k=1}^{n_s}  \|\EMST(\mathcal{D}_{\mathcal{S},k})\|$ instead of
$\sum\nolimits_{k=1}^{n_s}  \|\EMST(\mathcal{M}_{\mathcal{S},k})\|$, where $\EMST(\mathcal{M}_{\mathcal{S},k})$ denotes the Euclidean minimum spanning tree (EMST) over the spanning set
$\mathcal{U}_{\mathcal{S},k}$.
Since it holds that $\|\EMST(\mathcal{M}_{\mathcal{S},k})\| \geq \|\EMST(\mathcal{D}_{\mathcal{S},k})\|$,
we can obtain the following corollary.
\begin{cor}\label{corallary-n-s-tree-length}
For all multicast sessions $\mathcal{M}_{\mathcal{S},k}$ ($ 1 \leq k \leq n_s$),
it holds that for $n_d=o(\frac{n}{\log n})$,
\[\sum\nolimits_{k=1}^{n_s}  \|\EMST(\mathcal{M}_{\mathcal{S},k})\| =\Omega(n_s\cdot\sqrt{n_d \cdot n}).\]
\end{cor}

\section{Network Topology under Feasible Routings}
We introduce the Poisson Boolean percolation model to make preparations for computing the
upper bounds on the general multicast capacity.
\subsection{Poisson Boolean Percolation Model}

In a 2-dimensional Poisson Boolean model $\mathbb{B}( \lambda, r )$
\cite{meester1996cp},
 nodes are distributed in $\mathbb{R}^2$ according to a
p.p.p of intensity $\lambda$. Each node is associated with a closed disk of radius $r/2$.
Two disks are \emph{directly connected} if they overlap. Two disks
are \emph{connected} if there exist a sequence of directly
connected disks between them. Define a \emph{cluster} as a set of
disks in which any two disks are connected. Denote the set of all
clusters by $\mathscr{C}( \lambda, r )$. Let $|\mathcal{C}_i|$ denote the number of disks
in a cluster $\mathcal{C}_i \in \mathscr{C}( \lambda, r )$.
We can associate $\mathbb{B}( \lambda, r )$ with a graph $\mathcal{G}(
\lambda, r )$, called an \emph{associated graph}, by associating a
vertex with each node in $\mathbb{B}( \lambda, r )$ and an edge with
each direct connection in $\mathbb{B}( \lambda, r )$.
Two models $\mathbb{B}( \lambda, r
)$ and $\mathbb{B}( \lambda_0, r_0 )$ lead to the same associated
graph, namely $\mathcal{G}( \lambda, r )=\mathcal{G}( \lambda_0, r_0
)$  if $\lambda_0 {r_0}^2=\lambda {r}^2$. Then, the graph properties
of $\mathbb{B}( \lambda, r )$ only depend on the parameter $\lambda
{r}^2$, \cite{dousse2004cvc}.
Let $\mathcal{C}$ denote the cluster containing the
given node, the percolation probability is thus defined as $\Pr\nolimits_{\lambda,
r}[|\mathcal{C}|=\infty]$.
We call $\gamma_c$ the \emph{critical percolation threshold}
of Poisson Boolean model in $\mathbb{R}^2$  when
\begin{center}
$ \gamma_c= \sup\{\gamma:=\lambda \pi {r}^2 ~|~
\Pr\nolimits_{\lambda,
r}[|\mathcal{C}|=\infty]=0 \}$.
\end{center}
The exact value of $\gamma_c$ is still open. The
analytical results show that it is within the range $(0.7698\pi,
~3.372\pi)$~\cite{2007-kong-isit,{meester1996cp}}. In terms of the value of $\gamma=\lambda \pi {r}^2$, the \emph{subcritical phase} and \emph{supercritical phase} can be defined, which correspond to
the cases when $\gamma < \gamma_c$ and $\gamma > \gamma_c$, respectively. The following lemma will be used in our analysis.

\begin{lem}[\cite{meester1996cp,grossglauser:noc}]\label{lemma-Boolean-Model}
For a Poisson Boolean  model $\mathbb{B}( \lambda, r )$ in
$\mathbb{R}^2$,
there exists a value $\gamma_c$ in a square region
$\mathcal{R}(\lambda,n)=[0,\sqrt{n/\lambda}]^2$, as $n\to \infty$:
\begin{itemize}
  \item  if $\gamma=\lambda \pi {r}^2 < \gamma_c$, i.e., in the subcritical phase \cite{meester1996cp}, it holds that
 \begin{center}
  $\Pr[\sup\{|\mathcal{C}_i|~|~{\mathcal{C}_i \in \mathscr{C}( \lambda, r )}\}=O(\log n)]=1;$
 \end{center}
   \item  if $\gamma=\lambda \pi {r}^2 > \gamma_c$, i.e., in the supercritical phase \cite{meester1996cp}, there exists, \whp,
\emph{exactly} one \emph{giant cluster} (\emph{giant component})
$\mathcal{C}_i \in \mathscr{C}( \lambda, r )$ of size
$|\mathcal{C}_i|=\Theta(n)$.
\end{itemize}
\end{lem}

\subsection{Distance to Giant Component}

Connectivity is a necessary condition for a feasible routing scheme.
From \cite{gupta1998cpa,santi2003ctr},
the connectivity of a routing scheme for homogeneous random networks $\mathcal{N}(\lambda,n)$ can be ensured when the maximum link length can reach
$\Omega(\sqrt{\log n / \lambda})$.
More specifically,
by a geometric extension, we can obtain the following lemma
based on  Theorem 3.2 of \cite{gupta1998cpa}.
\begin{lem}\label{lemma-critical-connectivity}
In Poisson Boolean model $\mathcal{B}( \lambda,  r)$, with
\begin{center}
$\pi \cdot\lambda\cdot r^2=\log n +\varsigma(n)$,
\end{center}
all disks with radius $r$ are \emph{connected} with probability $1$ as $n\to \infty$ if and only if
$\varsigma(n)\to \infty$.
\end{lem}

From Lemma \ref{lemma-critical-connectivity}, we limit the nontrivial range of $r$ in
$[{\mathfrak{p}_c}/{\sqrt{\lambda}},\sqrt{\log n / \lambda}]$,
\ie, $ r:[{1}/{\sqrt{\lambda}},\sqrt{\log n / \lambda}] $. According to
Lemma \ref{lemma-Boolean-Model},  in the Poisson Boolean  model
$\mathcal{B}( \lambda,  r)$, there exists exactly one
giant component, denoted by  $\mathcal{C}( \lambda,  r)$,
with $|\mathcal{C}( \lambda,  r)|=\Theta(n)$. Note that
we take no account of the specific values of the involved constants, since they
have no impact on the order of our final results.

In Poisson Boolean  model $\mathcal{B}( \lambda,
r)$, for any node outside the giant cluster
$\mathcal{C}( \lambda,  r)$,  say an \emph{exterior node} $u \notin \mathcal{C}(
\lambda,  r)$,  we define the distance between $u$ and the
giant component by
\[\bar{l}_c(u)=\min\nolimits_{v\in \mathcal{C}( \lambda,  r)} |uv|.\]
Furthermore, we define the largest distance between exterior nodes and $\mathcal{C}( \lambda,  r)$ as
\[\bar{l}^{\mathrm{M}}_c\left[\mathcal{C}\left( \lambda,   r\right)\right]:=\max\nolimits_{u\in \mathcal{V}(n)-\mathcal{C}( \lambda, r)}\bar{l}_c(u),\]
where $\mathcal{V}(n)$ denotes the set of all nodes in $\mathcal{N}(\lambda, n)$.
Please see the illustration in Fig.\ref{Fig-outsidenodes}.

\begin{figure}[t]
\begin{center}
\scalebox{1}{\includegraphics{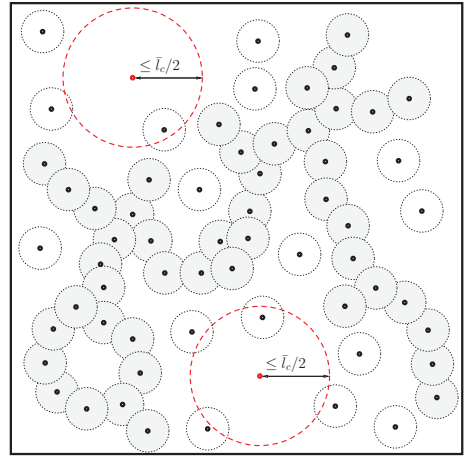}}
\end{center}
\caption{Distance
between an exterior node and the giant component.} \label{Fig-outsidenodes}
\end{figure}

From Lemma \ref{lemma-critical-connectivity}, there is no
node outside $\mathcal{C}( \lambda, r)$ when
\begin{center}
$\lambda \cdot
r^2=\frac{1}{\pi}\cdot (\log n +\varsigma(n))$ if
$\varsigma(n)\to \infty$.
\end{center}
 Then, we  only consider the case that
$\lambda \cdot r^2=o(\log n)$, \ie, $r=o(\sqrt{\log n /
\lambda})$. It holds that
\begin{center}
$    \bar{l}^{\mathrm{M}}_c\left[\mathcal{C}\left( \lambda,   r\right)\right] > r$, and $\bar{l}^{\mathrm{M}}_c\left[\mathcal{C}\left( \lambda,   r\right)\right]=o(\sqrt{\log n / \lambda})$.
\end{center}
Next, we give a useful result for computing the upper bounds on network capacity.
\begin{lem}
\label{lem-distance-outside} In Poisson Boolean  model $\mathcal{B}(
\lambda,  r)$ with $r=o(\sqrt{\log n / \lambda})$, it holds that
\begin{equation}\label{equ-key-upperbound-argue}
\lambda \cdot r \cdot \bar{l}^{\mathrm{M}}_c=\Omega(\log n), \whp,
\end{equation}
where $\bar{l}^{\mathrm{M}}_c:=\bar{l}^{\mathrm{M}}_c\left[\mathcal{C}\left( \lambda,   r\right)\right]$ for the sake of succinctness.
\end{lem}

Prior to proving Lemma \ref{lem-distance-outside}, we get the following lemma based on Corollary 1 of
\cite{dousse2006delay} by a geometric scaling method.
\begin{lem}[\cite{dousse2005distancegiant,dousse2006delay}]
For any exterior node, say $u\notin \mathcal{C}( \lambda, r)$, it
holds that for any $x\in [0, \sqrt{n/\lambda}]$,
\[
\lim_{n\to
\infty}\log \Pr \left[\bar{l}_c(u)> x \right]=- \lim_{n\to
\infty} \varepsilon \cdot \lambda \cdot
r  \cdot x ,
\]
where $\varepsilon>0$ is a constant.
\end{lem}

\begin{proof}[Proof of Lemma \ref{lem-distance-outside}]
Firstly, we give a bound on the probability of
event
\begin{center}
$\bar{E}(r)$: $\lambda \cdot r \cdot \bar{l}^{\mathrm{M}}_c=o(\log n)$ (a contradiction to Equation (\ref{equ-key-upperbound-argue})).
\end{center}
For any $u\notin \mathcal{C}( \lambda,
r)$, we define an event
\begin{center}
$\bar{E}_u(r)$: $\lambda \cdot r \cdot \bar{l}_c(u)=o(\log n)$.
\end{center}
Then, it follows that
\[
\Pr\left[\bar{E}(r)\right]= \Pr\left[\bigwedge\nolimits_{u\notin \mathcal{C}( \lambda,
r)} \bar{E}_u(r)\right]\leq \left(1- \frac{\varepsilon_1}{o(n)}\right)^{\varepsilon_2 n} \to 0,
\]
where $\varepsilon_1$ and $\varepsilon_2$ are some constants. Hence, the lemma is proved.
\end{proof}

\begin{figure}[t]
\begin{center}
 \scalebox{1}{\includegraphics{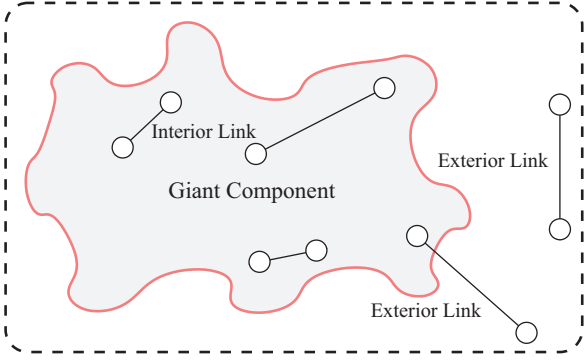}}
\end{center}
  \caption{Interior and Exterior Links.}\label{fig-cloud}
\end{figure}

%
%
%
%

\section{Upper Bounds on General Multicast Capacity}\label{sec-upper}

For any routing scheme, denote the maximum length (in the sense of order) of the
links  by $l_c$. According to \cite{02Gupta2000,15Franceschetti2007},
in the networking-theoretic scaling laws \cite{2012TC-MC},
under the premise of ensuring routing connectivity, long-distance communication is
not preferable, since the interference generated would preclude
too many nodes from communicating. The
optimal strategy is to confine to the nearest neighbor communication
and maximize the number of simultaneous transmissions, \ie, optimize the
spatial reuse, \cite{ozgur2007hca}.
From \cite{gupta1998cpa,santi2003ctr},
the routing connectivity of any scheme for homogeneous random networks can be ensured when the maximum link length is set to be
$\Omega(\sqrt{\log n / \lambda})$.
Then, we consider the range
$l_c:[{\mathfrak{p}_c}/{\sqrt{\lambda}},\sqrt{\log n / \lambda}]$,
\ie, $ l_c:[{1}/{\sqrt{\lambda}},\sqrt{\log n / \lambda}] $. From
Lemma \ref{lemma-Boolean-Model},  in the Poisson Boolean  model
$\mathcal{B}( \lambda,  {l_c})$, there exists exactly one
giant component, denoted by  $\mathcal{C}( \lambda,  {l_c})$,
with $|\mathcal{C}( \lambda,  {l_c})|=\Theta(n)$. Note that
we take no account of the specific values of the constants, for they
have no impact on the order of our final results.

Then, the links of any multicast scheme can be divided into two classes as
follows: A link is called an \emph{interior link}, if both endpoints are located in $\mathcal{C}( \lambda,  {l_c})$;
and it is called an \emph{exterior link}, otherwise.

In the Poisson Boolean  model $\mathcal{B}( \lambda,
{l_c})$, for any node outside the giant cluster
$\mathcal{C}( \lambda,  {l_c})$,  say $u \notin \mathcal{C}(\lambda,  {l_c})$, define the distance between $u$ and the
giant component by
\[\bar{l}_c(u)=\min\nolimits_{v\in \mathcal{C}( \lambda,  {l_c})} |uv|.\]
Furthermore, we define
\[\bar{l}^{\mathrm{M}}_c\left[\mathcal{C}\left( \lambda,   {l_c}\right)\right]:=\max\nolimits_{u\notin \mathcal{C}( \lambda,
 {l_c}/{2})}\bar{l}_c(u).\]
Please see the illustration in Fig.\ref{Fig-outsidenodes}.

We derive the upper bounds on multicast capacity by considering two types of links comprehensively.

\subsubsection{Inside a Giant Component}
All links inside  $\mathcal{C}( \lambda,  {l_c})$ have the
length of $\Theta(l_c)$. The upper bound on capacity of these links can be computed as
\[\mathbf{R}_{l_c}=\min\left\{1, B\log \left(1+\frac{l_c^{-\alpha}}{N_0} \right)\right\}=O(\min\{1, l_c^{-\alpha}\}).\]
Then, by combining with Lemma \ref{lem-basic-link-throughput}, we can obtain the following lemma.
\begin{lem}\label{lem-inside}
For any multicast scheme with the parameter $l_c$, the multicast throughput along
the links inside $\mathcal{C}( \lambda,  {l_c})$ is at most
of order $\Lambda_{l_c}=O\left(\frac{\min\{1,
l_c^{-\alpha}\}}{\mathbf{L}\left(n_s,  \frac{\sqrt{n}}{l_c\sqrt{n_d
\lambda}}\right)}\right)$.
\end{lem}
\begin{proof}
According to Lemma \ref{lem-K-EMST-length},
the length of any multicast tree is at least of order $\Omega(\sqrt{n_dn/\lambda})$.
Then, for a given sender of any links inside the giant component,
a multicast session passes through it with a probability of
\begin{center}
$\Omega\left(\min \left\{1, \frac{l_c\sqrt{n_dn/\lambda}}{n/\lambda}\right \}\right)$, \ie, $\Omega\left(\min\left\{1, \frac{l_c\sqrt{n_d \lambda}}{\sqrt{n}}\right\}\right)$.
\end{center}
By Lemma \ref{lem-basic-link-throughput}, the proof is completed.
\end{proof}

\subsubsection{Outside a Giant Component}

Based on Lemma \ref{lem-distance-outside}, we have,

\begin{lem}\label{lem-outside}
For any multicast scheme with $l_c$, the multicast throughput along
the links between $\mathcal{C}( \lambda,  {l_c})$ and the
nodes outside is at most of order
$\Lambda_{\bar{l}^{\mathrm{M}}_c}=O\left(\frac{\min\{1, (\frac{\lambda}{\log
n})^{\alpha/2}\}}{\mathbf{L}(n_s,  \frac{ n\sqrt{ \lambda} \cdot
l_c}{n_d \cdot \sqrt{\log n}})}\right)$.
\end{lem}
\begin{proof}
Since there must be a link outside the giant component with the length of $\sqrt{\log n / \lambda}$,
the link capacity is bounded by
\begin{eqnarray*}
  \mathbf{R}_{\bar{l}^{\mathrm{M}}_c} &=& \min\left\{1, B\log \left(1+\frac{(\sqrt{\log n/ \lambda})^{-\alpha}}{N_0}\right)\right\} \\
  &=& O\left(\min\left\{1,  \left(\frac{\lambda}{\log n}\right)^{\alpha/2}\right\}\right).
\end{eqnarray*}
From Lemma \ref{lem-distance-outside}, $\bar{l}^{\mathrm{M}}_c=\Omega\left(\frac{\log n}{\lambda \cdot l_c}\right)$. It implies that  $\bar{l}^{\mathrm{M}}_c=\Omega(\sqrt{\log n / \lambda})$
because $ l_c:[{1}/{\sqrt{\lambda}},\sqrt{\log n/\lambda}] $. The
probability that a multicast session passes through such a link is
of
\[\Omega\left(\min\left\{1, \frac{n_d \cdot \bar{l}^{\mathrm{M}}_c  \bar{l}^{\mathrm{M}}_c \cdot \sqrt{\lambda}
}{n\cdot \sqrt{\log n }}\right\}\right).\]
 By Lemma \ref{lem-basic-link-throughput}, the proof is completed.
\end{proof}

By combining Lemma  \ref{lem-inside} and Lemma \ref{lem-outside}, we
finally obtain Theorem \ref{thm-upper-bound}.

\section{Conclusion and Discussion}\label{section-Conclusion and Future Work}

We derive the general upper bounds on the capacity  for
random wireless networks with a general node density.
When the general results are specialized  to
the well-known random dense and extended networks,
we show that our results close the open gaps between the upper and the lower bounds
on the multicast capacity for both networks.

{\small
\bibliographystyle{IEEEtran}
\bibliography{Capacity}
}

\newpage

\clearpage

\begin{flushleft}
{\huge Supplementary File}
\end{flushleft}

{\numberwithin{equation}{section}
{\numberwithin{lem}{section}
{\numberwithin{thm}{section}
{\numberwithin{defn}{section}
{\numberwithin{figure}{section}
{\numberwithin{table}{section}

\appendices

\section{Lower Bounds on General Multicast Capacity}
\label{sec-lower}
We design two general multicast schemes by using two types of hierarchical backbones systems in a well-integrated manner.
One hierarchical backbones system consists of the \emph{highways} and \emph{ordinary arterial roads};
the other is composed of the \emph{highways} and \emph{parallel arterial roads}.
Combining the achievable throughputs under our two schemes and other two schemes \cite{2012-winet-asymptotic} that are respectively based only on
ordinary arterial roads and parallel arterial roads,
we derive the optimal throughput as  the lower bounds on general multicast capacity
according to different ranges of parameters.

For the sake of succinctness, we first introduce a notion
called \emph{scheme lattice} from \cite{2012TC-MC}.
\begin{defn}[Scheme Lattice, \cite{2012TC-MC}]\label{defn-Scheme-Lattice}
Divide the deployment region
$\mathcal{R}(\lambda,n)=[0,\sqrt{n/\lambda}]^2$ into a lattice
consisting of square cells of side length $b$, we call the lattice
\emph{scheme lattice} and denote it by
$\mathbb{L}(\sqrt{n/\lambda},b,\theta)$, where $\theta \in [0,
{\pi}/{4}]$ is the minimum angle
 between the sides of the deployment region and produced cells.
\end{defn}

In our multicast schemes, the backbones of routing comprise two levels:
 \emph{highway system} and \emph{arterial road system}.
The highway system based on \emph{bond percolation theory} \cite{grimmett1999p} was originally proposed in
\cite{15Franceschetti2007}; and the connectivity-based arterial road system was devised in \cite{2012-winet-asymptotic}.
The main novelty of schemes in this work is the adoption of these two types of backbone systems in an integrated manner.
For the sake of completeness, we introduce concisely the construction procedures of these backbone systems,
and extend some relevant results in \cite{15Franceschetti2007} and \cite{2012-winet-asymptotic}
into the scenarios with general node density by a  geometric scaling, respectively.

\subsection{Highway System}

\subsubsection{Construction of highway system}
The highways are built based on scheme
lattice $\mathbb{L}(\sqrt{n/\lambda}, \sqrt{c^2/\lambda},  {\pi}/{4})$, as illustrated in
Fig.\ref{Fig-highways}. Then, there are $m^2$ cells, where $m = \left\lceil  \sqrt{n}/\sqrt{2}c \right\rceil^2 $.
A cell  is non-empty (open)
with the probability of
$p\to 1-\exp({-c^2})$, as $n\to \infty$, independently from each other.
Based on $\mathbb{L}(\sqrt{n/\lambda}, \sqrt{c^2/\lambda}, {\pi}/{4})$, draw a
horizontal edge across half of the squares, and a vertical edge
across the others, to obtain a new  lattice as described in Fig.\ref{Fig-highways}.
 An edge  $\hbar$ in the new lattice is
\emph{open} if the cell
crossed  by $\hbar$ is \emph{open}, and call a path comprised of
edges in the new lattice (Fig.\ref{Fig-highways}) \emph{open} if it
contains only open edges. Based on an open path penetrating the deployment region, as illustrated
in Fig.\ref{Fig-highways}, we choose a node from each cell
in $\mathbb{L}(\sqrt{n/\lambda}, \sqrt{c^2/\lambda}, {\pi}/{4})$ corresponding to the
edge of open path, call this node \emph{highway-station},
connect a pair of highway-stations in two
adjacent cells,  and finally obtain a crossing path, and call it \emph{highway}, as in Fig.\ref{Fig-highways}.

\begin{figure}[h]
\begin{flushright}
\scalebox{1}{\includegraphics{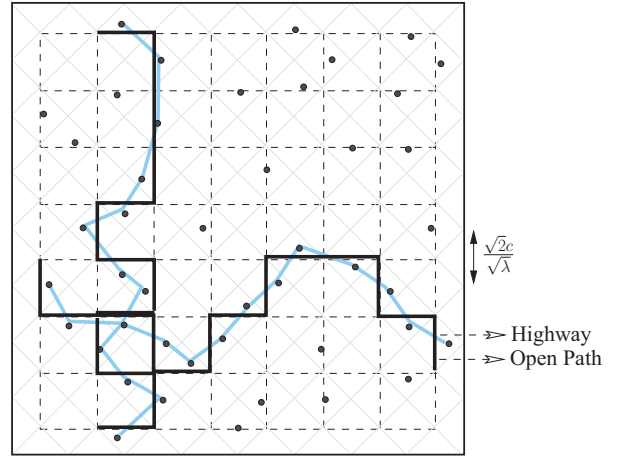}}
\end{flushright}
 \caption{Building Horizontal and Vertical Highways, \cite{2011-cheng-general-capacity}.
}
\label{Fig-highways}
\end{figure}

\begin{table}[t] \renewcommand{\arraystretch}{1.9}
\caption{Notations used in this paper.} \label{tab-notions}
 \centering
 \scalebox{0.98}{\begin{tabular}{c|c} \hline Notation & Meaning\\
  \hline \hline
  $\mathbb{L}(\cdot,\cdot,\cdot)$ &
   Scheme Lattice (Definition \ref{defn-Scheme-Lattice})\\
\hline
AR  & Arterial Road \\
\hline
AR-cell & The cell in $\mathbb{L}(\sqrt{n/\lambda}, 3\sqrt{\log n/\lambda},0)$ \\
\hline
Station-cell & The square cell centered at AR-cell of area  $\frac{4\log n}{\lambda}$, Fig.\ref{Fig-AR-system}. \\
\hline
PA-cell & Parallel Assignment Cell-subsquare in  AR-cell of area $\frac{9}{2\lambda}$. \\
\hline
O-AR, P-AR & Ordinary Arterial Road, Parallel Arterial Road \\
\hline
  O-AP, P-AP & Ordinary Access Path, Parallel Access Path  \\
\hline
$\mathcal{U}_{k}$     & Spanning Set of Multicast Session $\mathcal{M}_k$\\
\hline
$\mathbf{S}_{\mathrm{o}}(v)$& The entry point from   node $v$ to an assigned O-AR\\
\hline
$\mathbf{S}_{\mathrm{p}}(v)$& The entry point from  node $v$ to an assigned P-AR \\
\hline
$\EST(\mathcal{U}_{k})$& An Euclidean Spanning Tree  of Multicast Session $\mathcal{M}_k$ \\
\hline
$\mathbb{M}_{\mathrm{o}}$  & Scheme based on only O-AR system\\
\hline
$\mathbb{M}_{\mathrm{p}}$  & Scheme based on only P-AR system\\
\hline
$\mathbb{M}_{\mathrm{o}\&\mathrm{h}}$  & Scheme based on both O-AR and highway system\\
\hline
$\mathbb{M}_{\mathrm{p}\&\mathrm{h}}$  & Scheme based on both P-AR  and highway system\\
\hline
  \end{tabular}
  }
  \end{table}

For a given constant $\kappa >0$, partition the scheme lattice
$\mathbb{L}(\sqrt{n/\lambda}, \sqrt{c^2/\lambda}, {\pi}/{4})$ into horizontal (or
vertical) rectangle slabs of size $m \times  \kappa \log m$ (or $  \kappa \log m  \times  m$), denoted
by $\mathcal{R}^{\mathrm{H}}_{i}$ (or $\mathcal{R}^{\mathrm{V}}_{i}$),
where $m=\frac{\sqrt{n}}{\sqrt{2}c}$.
 Denote the
number of disjoint horizontal (or vertical) highways within
$\mathcal{R}^{\mathrm{H}}_i$ (or $\mathcal{R}^{\mathrm{V}}_i$) by $N^{\mathrm{H}}_i$ (or
$N^{\mathrm{V}}_i$). The next lemma follows.
\begin{lem}  (\cite{15Franceschetti2007}) \label{lemma-density-highway}
For any $\kappa$ and $p\in (5/6, 1)$ satisfying $2+\kappa
\log(6(1-p))<0$, there exists a $\eta =\eta(\kappa, p)$ such that
\[
  \lim\limits_{m \to \infty} \Pr[N^h \geq \eta \log m]=1,
  \lim\limits_{m \to \infty} \Pr[N^v \geq \eta \log m]=1,
\]
where $N^{\mathrm{H}}=\min_i N^{\mathrm{H}}_i$ and $N^{\mathrm{V}}=\min_i N^{\mathrm{V}}_i$.
\end{lem}

\subsubsection{Transmission scheduling for highway system:} The highways can be scheduled by a 9-TDMA  scheme based on scheme lattice
$\mathbb{L}(\sqrt{n/\lambda}, \sqrt{c^2/\lambda}, {\pi}/{4})$, \cite{15Franceschetti2007}. By a similar to Theorem 3 in
\cite{15Franceschetti2007}, we can prove that all highways can sustain w.h.p. the rate of
order $\Omega(1)$.

\subsection{Arterial Road  (AR) System}
We introduce two types of arterial road  (AR) systems from \cite{2012-winet-asymptotic}: \emph{ordinary
arterial road system} and \emph{parallel arterial road system},
which perform better than the other according to the different density
$\lambda$. Both AR systems are constructed based on the scheme
lattice $\mathbb{L}(\sqrt{n/\lambda}, 3\sqrt{\log n/\lambda},0)$.
Then, there are
$\frac{n}{9\log n}$ cells in $\mathbb{L}(\sqrt{n/\lambda},
3\sqrt{\log n/\lambda},0)$, called \emph{AR-cells}. Denote each row
(or column) by $\mathcal{\tilde{R}}_i^h$ (or
$\mathcal{\tilde{R}}_i^v$), where $i\in[1,
\frac{\sqrt{n}}{3\sqrt{\log n}}]$.
Then, from Lemma \ref{lem-poisson-trial-tail},
for all $\frac{n}{9\log n}$ AR-cells,
the number of nodes is \whp within $[\frac{9}{2}\log n, 18\log n]$.

Firstly, we introduce  the \emph{ordinary arterial road system}.

\subsubsection{Ordinary Arterial Road System (O-AR system)}

The \emph{ordinary arterial road system} can be obtained by choosing randomly one
node from each cell, called \emph{ordinary AR-station}, and connecting these stations in edge-adjacent cells.
Then, we have
\begin{lem}[\cite{2012-winet-asymptotic}]\label{lemma-rate-ordinary}
By a $9$-TDMA scheme,
each ordinary arterial road in O-AR system  can sustain a rate of
$$\mathbf{R}_{\mathrm{O-AR}}(\lambda, n)=\left\{
\begin{array}{lll}
\Theta(\frac{\lambda^{\frac{\alpha}{2}}}{(\log n)^{\frac{\alpha}{2}}}) & \mbox{when} &  \lambda :[1, \log n] \\
\Theta(1) & \mbox{when} &   \lambda :[\log n, n] \\
  \end{array} \right.$$
 \end{lem}

\begin{figure}[t]
\scalebox{1}{\includegraphics{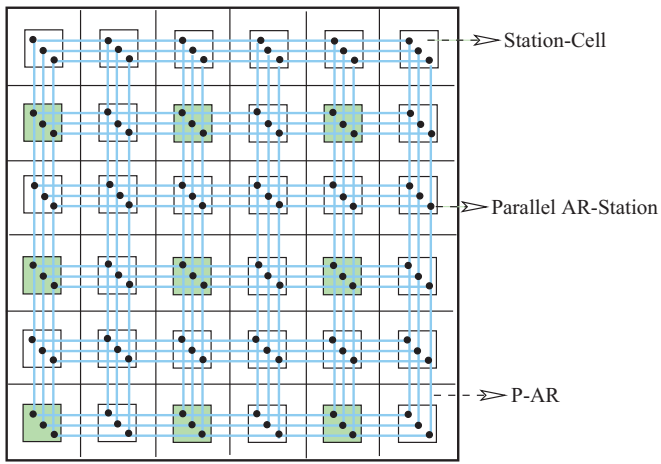}}
\caption{Parallel
Arterial Roads.
The shaded station-cells can
  be scheduled
simultaneously.
  In any time slot, there are
$2\log n$ concurrent links initiated from every activated
station-cell, \cite{2012-winet-asymptotic}.} \label{Fig-AR-system}
\end{figure}

Next, we introduce the \emph{parallel arterial road system}.

\subsubsection{Parallel Arterial Road System  (P-AR system)}

In the center of each AR-cell,
we set a smaller square of side length $2 \sqrt{\log n /\lambda}$,
as illustrated in Fig.\ref{Fig-AR-system},  call it
\emph{station-cell}. Then, by
Lemma \ref{lem-poisson-trial-tail}, we can prove that
for all station-cells, there are, \whp,  at least
 $2\log n$ nodes.

The
\emph{horizontal} arterial roads in $\mathcal{\tilde{R}}_i^h$ is constructed by using the
following operations: Firstly, for all $\frac{\sqrt{n}}{3\sqrt{ \log n}}$
station-cells in $\mathcal{\tilde{R}}_i^h$, choose $2\log n$ nodes from each station-cell,
called \emph{parallel AR-stations}. Secondly,  connect those parallel AR-stations in the station-cells contained    in the edge-adjacent AR-cells in a one-to-one pattern, as illustrated in Fig.\ref{Fig-AR-system}.
In a
similar way, we can construct the \emph{vertical} arterial roads.
We say that two  arterial roads are \emph{disjoint} if  no
station is shared by them. According to the procedure of
construction above, there are $2\log n$ disjoint horizontal (or vertical)
arterial roads in every row (or column) of $\mathbb{L}(\sqrt{n/\lambda}, 3\sqrt{\log n/\lambda},0)$.

A $4$-TDMA scheme, as depicted in Fig.
\ref{Fig-AR-system}, is adopted to schedule
arterial roads.
The main technique
called \emph{parallel transmission scheduling} is: Instead of
scheduling only one link in each activated station-cell (or cell) in each time slot,
we consider scheduling $2\log n$ links   initiating from the same
station-cell (or cell) together.
It can be proven that
 this
modification  increases the total throughput for each cell by order
of $\Theta(\log n)$, compared with only scheduling one link in each
cell.
\begin{lem}[\cite{2012-winet-asymptotic}]\label{lemma-P-AR-rate}
Each P-AR  can sustain a rate of
$$\mathbf{R}_{\mathrm{P-AR}}(\lambda, n)=\left\{
\begin{array}{lll}
\Theta(\frac{\lambda^{\frac{\alpha}{2}}}{(\log n)^{\frac{\alpha}{2}}}) & \mbox{when} &  \lambda :[1, (\log n)^{1-\frac{2}{\alpha}}] \\
\Theta(\frac{1}{\log n} ) & \mbox{when} &   \lambda :[(\log n)^{1-\frac{2}{\alpha}}, n] \\
  \end{array} \right.$$
\end{lem}

\subsection{Access Paths}
We assign nodes to the specific arterial roads by now.
Next, we devise the
\emph{access path}, including \emph{draining paths} and \emph{delivering paths}, for every node to the arterial road system.

\subsubsection{Access Paths to O-AR System (O-APs)}
We call those links, along which the nodes outside drain the packets to  O-AR system or the stations in O-AR system deliver the
packets to the nodes outside,
\emph{ordinary access paths} (O-APs).

For every node outside ordinary arterial roads, say $v$,
it drains (or receives) data packets to (or from) the ordinary AR-station in the AR-cell containing $v$,
denoted by  $\mathbf{S}_{\mathrm{o}}(v)$,   by a single hop called
\emph{ordinary draining path} (or \emph{ordinary delivering path}).

A 4-TDMA scheme based on $\mathbb{L}(\sqrt{n/\lambda}, 3\sqrt{\log n/\lambda},0)$
is adopted to schedule the O-APs.
Each slot can be further divided into $8\log n$ subslots, ensuring that every link contained in each  AR-cell
can be scheduled once in a period of $4\times 8\log n$ subslots. Then, it follows that
\begin{lem}[\cite{2012-winet-asymptotic}]\label{lemma-ordinary-access-rate}
The rate of each ordinary access path, including ordinary draining path and ordinary delivering path,
can \emph{also} be sustained of
$$\mathbf{R}_{\mathrm{O-AR}}(\lambda, n)=\left\{
\begin{array}{lll}
\Theta(\frac{\lambda^{\frac{\alpha}{2}}}{(\log n)^{\frac{\alpha}{2}}}) & \mbox{when} &  \lambda :[1, \log n] \\
\Theta(1) & \mbox{when} &   \lambda :[\log n, n] \\
  \end{array} \right.$$
\end{lem}

\subsubsection{Access Paths to P-AR System (P-APs)}
We call those links, along which the nodes outside drain the packets to  P-AR system or the stations in P-AR system deliver the
packets to the nodes outside,
\emph{parallel access paths} (P-APs).

For every node outside parallel arterial roads, say $v$, where $v \in  \mathcal{\tilde{R}}_j^v$ and $v \in  \mathcal{\tilde{R}}_i^h$,
it drains the data packets into a  parallel AR-station located in
the adjacent AR-cell in $\mathcal{\tilde{R}}_j^v$, denoted by  $\mathbf{S}_{\mathrm{p}}(v)$,
 by a single hop called \emph{parallel draining path} (Please see the illustration in Fig.\ref{fig-accessing-path}(a)); and receives the packets from the station, located in
the adjacent AR-cell in $\mathcal{\tilde{R}}_i^h$,
of
a specific arterial road by a single hop called \emph{parallel delivering path}
(Please see the illustration in Fig.\ref{fig-accessing-path}(b)).
Specifically, each AR-cell is further divided into $2\log n$ subsquares, called \emph{parallel assignment cell} (PA-cell), of area $\frac{9\log n/\lambda}{2\log n}=\frac{9}{2\lambda}$.
Connect all nodes in the same PA-cell with the same P-AR station in the adjacent AR-cell to build the
P-APs.

A 2-TDMA scheme is capable to
schedule the draining paths (delivering paths, resp.) except those
initiating  from (terminating to, resp.) nodes in
$\mathcal{\tilde{R}}_{\delta}^h$ ($\mathcal{\tilde{R}}_{\delta}^v$,
resp.), where $\delta=\frac{\sqrt{n}}{3\sqrt{\log n}}$, and use an
additional 1-TDMA scheme to schedule other draining paths
(delivering paths, resp.). Please see the illustrations in
Fig.\ref{fig-accessing-path}(a) and Fig.\ref{fig-accessing-path}(b).
Then, it follows that
\begin{lem}[\cite{2012-winet-asymptotic}]\label{lemma-parallel-access-rate}
The rate of each parallel access path, including parallel draining and parallel delivering paths,
 can also be sustained of
$\mathbf{R}_{\mathrm{P-AR}}(\lambda, n)$.
\end{lem}

\subsection{Multicast Routing Schemes}

\subsubsection{Euclidean Spanning Tree}
We recall a result from \cite{LiTON08}.
\begin{lem}[\cite{LiTON08} ]\label{lemma-EMST-upperbound}
 For any spanning set $ \mathcal{U}_{k} $ consisting of $ n_{d}+ 1$   nodes    placed in a square $\mathcal{R}=[0,\mathfrak{a}]^2$,  the length
 of \emph{Euclidean spanning tree}
$\EST(\mathcal{U}_{k})$ obtained by
the algorithm in \cite{LiTON08}
is at most of $ 2\sqrt{2} \cdot\sqrt
{n_d+1}\cdot \mathfrak{a}$.
\end{lem}

Then, for any multicast session $\mathcal{M}_k$, based on its spanning set $\mathcal{U}_k$, we build an
Euclidean spanning tree, denoted by $\EST(\mathcal{U}_{k})$.
Denote the set of all edges of $\EST(\mathcal{U}_{k})$ by $\mathcal{E}_k$.

\subsubsection{Assignment of Backbones}
Now, we determine which backbones, including highway and AR,
can be used by a specific communication-pair, \ie, a link $u\to v \in \mathcal {E}_k$.

\textbf{Assignment of Arterial Roads:} Denote the vertical O-AR (or
P-AR) passing through the ordinary (or parallel) AR-station
$\mathbf{S}_\mathrm{o}(u)$ (or $\mathbf{S}_\mathrm{p}(u)$) by
$\mathbf{AR}^\mathrm{V}_\mathrm{o}(u)$ (or
$\mathbf{AR}^\mathrm{V}_\mathrm{p}(u)$); and denote the horizontal
O-AR (or P-AR) passing through the ordinary (or parallel) AR-station
$\mathbf{S}_\mathrm{o}(v)$ (or $\mathbf{S}_\mathrm{p}(v)$) by
$\mathbf{AR}^\mathrm{H}_\mathrm{o}(v)$ (or
$\mathbf{AR}^\mathrm{H}_\mathrm{p}(v)$).

\textbf{Assignment of Highways:} Recall from Lemma
\ref{lemma-density-highway} that in each horizontal (or vertical)
rectangle slab $\mathcal{R}^{\mathrm{H}}_{i}$ (or
$\mathcal{R}^{\mathrm{V}}_{i}$) of area $\sqrt{n} \times  \kappa
\sqrt{2c} \cdot\log \frac{\sqrt{n}}{\sqrt{2c}}$ (or $   \kappa
\sqrt{2c} \cdot\log \frac{\sqrt{n}}{\sqrt{2c}}  \times  \sqrt{n}$),
there are at least $\eta \cdot \log \frac{\sqrt{n}}{\sqrt{2}c}$
horizontal (or vertical) highways. Divide further each horizontal
(or vertical) slab into horizontal (or vertical) slice of area
$\sqrt{n} \times  \frac{\kappa \sqrt{2c}}{\eta}$ (or $ \frac{\kappa
\sqrt{2c}}{\eta}  \times  \sqrt{n}$). Choose any $\eta \cdot \log
\frac{\sqrt{n}}{\sqrt{2}c}$ highways from each slab, and define an
arbitrary bijection from those highways to the slices. For any node
$u$ located in a horizontal slice $\mathbf{Slice}^{\mathrm{H}}_{j}$
(or vertical slice $\mathbf{Slice}^{\mathrm{V}}_{j}$), the packets
initiating from $u$ and terminating to $v$ are assigned to the
horizontal highway $\mathbf{H}^{\mathrm{H}}(u)$ and vertical highway
$\mathbf{H}^{\mathrm{V}}(v)$ that are mapped to the slices
$\mathbf{Slice}^{\mathrm{H}}_{j}$ and
$\mathbf{Slice}^{\mathrm{V}}_{j}$, respectively.

\subsubsection{Multicast Routing Schemes}

For each multicast session $\mathcal{M}_k$ with an Euclidean spanning tree  $\EST(\mathcal{U}_{k})$,
we build two types of  multicast routing trees by  two  corresponding schemes,
denoted by
$\mathbb{M}_{\mathrm{o}\&\mathrm{h}}$ and $\mathbb{M}_{\mathrm{p}\&\mathrm{h}}$, as described in Table.\ref{tab-notions}.

For each edge $u\to v \in \mathcal {E}_k$:
\begin{itemize}
  \item [Under $\mathbb{M}_{\mathrm{o}\&\mathrm{h}}$,]
  $u$ drains the packets into the ordinary AR-station $\mathbf{S}_\mathrm{o}(u)$ along a specific O-AP;
the packets are transported along the vertical ordinary AR
$\mathbf{AR}^\mathrm{V}_\mathrm{o}(u)$  to the assigned horizontal
highway $\mathbf{H}^{\mathrm{H}}(u)$;
the packets are carried along $\mathbf{H}^{\mathrm{H}}(u)$ and then the vertical highway $\mathbf{H}^{\mathrm{V}}(v)$;
the packets are transported along  $\mathbf{AR}^\mathrm{H}_\mathrm{o}(v)$ to  the ordinary AR-station $\mathbf{S}_\mathrm{o}(v)$; and
this station delivers the packets to $v$.
  \item [Under $\mathbb{M}_{\mathrm{p}\&\mathrm{h}}$,]
  $u$ drains the packets into the parallel AR-station $\mathbf{S}_\mathrm{p}(u)$ along a specific P-AP;
the packets are transported along the vertical parallel AR
$\mathbf{AR}^\mathrm{V}_\mathrm{p}(u)$  to the assigned horizontal
highway $\mathbf{H}^{\mathrm{H}}(u)$;
the packets are carried along $\mathbf{H}^{\mathrm{H}}(u)$ and then the vertical highway $\mathbf{H}^{\mathrm{V}}(v)$;
the packets are transported along  $\mathbf{AR}^\mathrm{H}_\mathrm{p}(v)$ to  the parallel AR-station $\mathbf{S}_\mathrm{p}(v)$; and
this station delivers the packets to $v$.
\end{itemize}
When all links in $\mathcal {E}_k$ are checked, merge the same edges (hops) and remove the circles
that cannot break the connectivity of $\EST(\mathcal{U}_{k})$. Finally, we
obtain the corresponding multicast routing trees.

\begin{figure}[t]
\begin{flushleft}
\begin{tabular}{cc}
\scalebox{0.9}{\includegraphics{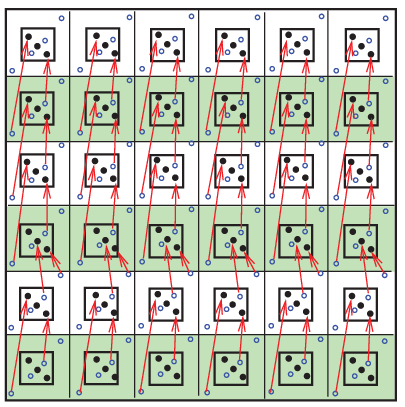} }&
\scalebox{0.9}{\includegraphics{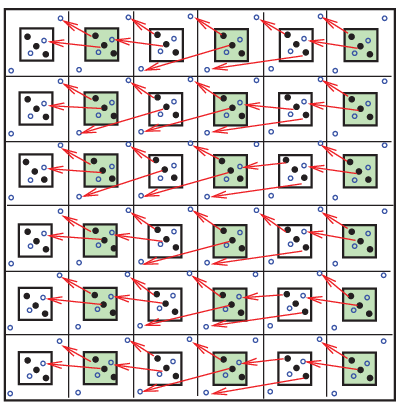}}
 \\
(a) Parallel Draining Paths & (b) Parallel Delivering Paths
\end{tabular}
\end{flushleft}
\caption{(a) The shaded cells can be scheduled simultaneously. All
draining paths except those initiating from nodes in
$\mathcal{R}_{\delta}^h$, where $\delta=\frac{\sqrt{n}}{3\sqrt{\log
n}}$, can be scheduled once in $2\times \frac{16\log n}{2\log n}=16$
time slots. In each slot, $2\log n$ links can be scheduled
simultaneously. Here, $16\log n$ is the maximum number of nodes in
each cell, and  $2\log n$ is the number of stations in each cell. In
addition, the nodes in $\mathcal{R}_{\delta}^h$ drain packets to the
stations in $\mathcal{R}_{\delta-1}^h$, and those access paths can
be scheduled by additional $\frac{16\log n}{2\log n}=8$ time slots.
(b) The shaded station-cells can
  be scheduled
simultaneously.
 All delivering paths except those terminating to nodes in $\mathcal{R}_{\delta}^v$,
can be scheduled once in $2\times \frac{16\log n}{2\log n}=16$ time
slots. In each slot, $2\log n$ links can be scheduled
simultaneously. In addition, the nodes in $\mathcal{R}_{\delta}^v$
receive packets from the stations in $\mathcal{R}_{\delta-1}^v$, and
those access paths can be scheduled by  additional $\frac{16\log
n}{2\log n}=8$ time slots.} \label{fig-accessing-path}
 \end{figure}

\subsection{Achievable Multicast Throughput}
Let $\mathbb{M}_{\mathrm{o}}$ and $\mathbb{M}_{\mathrm{p}}$ denote respectively the schemes only using O-AR system and using
P-AR system, \cite{2012-winet-asymptotic}.
By deriving the optimal throughput based on these four schemes
$\mathbb{M}_{\mathrm{o}}$, $\mathbb{M}_{\mathrm{p}}$,
$\mathbb{M}_{\mathrm{o}\&\mathrm{h}}$, and $\mathbb{M}_{\mathrm{p}\&\mathrm{h}}$,
we can obtain Lemma \ref{thm-achievable-MT-all-schemes}.

According to \cite{2012-winet-asymptotic}, under schemes $\mathbb{M}_{\mathrm{o}}$ and $\mathbb{M}_{\mathrm{p}}$,
the multicast throughputs can be respectively achieved as
$$\Lambda_{\mathrm{o}}(\lambda, n)=\frac{\mathbf{R}_{\mathrm{O-AR}}(\lambda, n)}{\mathbf{L}(n_s,\frac{1}{\mathbf{p}_\mathrm{o}})},~\Lambda_{\mathrm{p}}(\lambda, n)=\frac{\mathbf{R}_{\mathrm{P-AR}}(\lambda, n)}{\mathbf{L}(n_s,\frac{1}{\mathbf{p}_\mathrm{p}})}.$$

Next, we analyze our new schemes, \ie, $\mathbb{M}_{\mathrm{o}\&\mathrm{h}}$ and $\mathbb{M}_{\mathrm{p}\&\mathrm{h}}$.

\subsubsection{Scheme Using  Both the O-AR and Highway Systems, $\mathbb{M}_{\mathrm{o}\&\mathrm{h}}$}
The routing realization of any link in $\mathcal{E}_k$, say $u\to v$,
can be divided into three phases:
ordinary access path (O-AP) phase during which the packets are drained into O-ARs (or delivered from O-ARs) via O-APs,
ordinary arterial Road (O-AR) phase during which the packets are drained  into highways (or delivered from highways)
along O-ARs,
and highway phase during which the packets are transported along the highways.
Consider the throughput during all three phases, we can obtain the multicast throughput
under the scheme $\mathbb{M}_{\mathrm{o}\&\mathrm{h}}$ according to bottleneck principle.
\begin{lem}\label{lem-throughput-ordinary-highway}
Under the multicast scheme $\mathbb{M}_{{\mathrm{o}\&\mathrm{h}}}$, the multicast throughput is achieved as
$$\Lambda_{{\mathrm{o}\&\mathrm{h}}}(\lambda, n)=\min\left\{\frac{\mathbf{R}_{\mathrm{O-AR}}(\lambda,
n)}{\mathbf{L}(n_s,
\frac{1}{\mathbf{p}_{\mathrm{oh},\mathrm{O-AR}}})},
\frac{1}{\mathbf{L}(n_s,\frac{1}{\mathbf{p}_{\mathrm{oh},\mathrm{H}}})}\right\}.$$
\end{lem}
\begin{proof}
Since O-APs can sustain the same rate (in order sense) as that of
O-ARs, and the maximum burden of O-APs is necessarily not more than
that of O-ARs, we neglect the analysis of O-AP phase, and only
consider the O-AR phase and highway phase.

\emph{O-AR Phase:} For any AR-station, say
$\mathbf{S}_{\mathrm{oh},\mathrm{O-AR}}$, define an event
$E_k(\mathbf{S}_{\mathrm{oh},\mathrm{O-AR}})$ for  $\mathcal{M}_k$:
$\mathcal{M}_k$ shares the bandwidth of the link of an AR initiating
from the station $\mathbf{S}_{\mathrm{oh},\mathrm{O-AR}}$ during the
O-AR phase of  multicast scheme
$\mathbb{M}_{\mathrm{o}\&\mathrm{h}}$. Clearly, if
$E_k(\mathbf{S}_{\mathrm{oh},\mathrm{O-AR}})$ happens, then there is
an edge $u\to v \in \mathcal{E}_k$ such that the event
$E_{k;u,v}(\mathbf{S}_{\mathrm{oh},\mathrm{O-AR}})$ occurs, where
the event $E_{k;u,v}(\mathbf{S}_{\mathrm{oh},\mathrm{O-AR}})$ is
defined as: the routing path of $u \to v$ under the scheme
$\mathbb{M}_{\mathrm{o}}$ passes through
$\mathbf{S}_{\mathrm{oh},\mathrm{O-AR}}$. Obviously,
$E_k(\mathbf{S}_{\mathrm{oh},\mathrm{O-AR}})=\bigcup_{uv\in \Pi_k}
E_{k;u,v}(\mathbf{S}_{\mathrm{oh},\mathrm{O-AR}})$. Then,
\begin{eqnarray*}
  \Pr(E_k(\mathbf{S}_{\mathrm{oh},\mathrm{O-AR}})) &\leq& n_d \cdot  \Pr(E_{k;u,v}(\mathbf{S}_{\mathrm{oh},\mathrm{O-AR}})) \\
 &\leq & n_d \cdot \frac{6\sqrt{\log n/\lambda} \cdot \sqrt{2/\lambda}c \log \frac{\sqrt{n}}{\sqrt{2}c}}{n/\lambda} \\
 &\leq &\frac{6n_d \cdot (\log n)^{3/2}}{n} 
\end{eqnarray*}
Define $\mathbf{p}_{\mathrm{oh},\mathrm{O-AR}}=\min\{\frac{6n_d
\cdot (\log n)^{3/2}}{n},1\}$. Then, according to Lemma
\ref{lem-basic-link-throughput}, we obtain that the throughput
during the AR phase of scheme $\mathbb{M}_{\mathrm{o}\&\mathrm{h}}$
is achieved as  $\frac{\mathbf{R}_{\mathrm{O-AR}}(\lambda,
n)}{\mathbf{L}(n_s,
\frac{1}{\mathbf{p}_{\mathrm{oh},\mathrm{O-AR}}})}$.

\emph{Highway Phase:}
The routing realization of any multicast
session  $\mathcal{M}_k$ passes through a station during the highway
phase with the probability at most of
\begin{center}
$\mathbf{p}_{\mathrm{oh},\mathrm{H}}=\left\{ \begin{array}{ll}
\Theta(\sqrt{\frac{n_d}{n}}) & \mathrm{when}~n_d:[1,\frac{n}{(\log n)^2}] \\
\Theta(\frac{n_d\log n}{n}) & \mathrm{when}~n_d:[\frac{n}{(\log n)^2},\frac{n}{\log n}]\\
\Theta(1) & \mathrm{when}~n_d:[\frac{n}{(\log n)},n]\\
  \end{array} \right.     $
\end{center}
From Lemma
\ref{lem-basic-link-throughput}, we get that the throughput during
highway phase of multicast scheme
$\mathbb{M}_{\mathrm{o}\&\mathrm{h}}$ can be achieved as
$\frac{\mathbf{R}_{\mathrm{H}}(\lambda,
n)}{\mathbf{L}(n_s,\frac{1}{\mathbf{p}_{\mathrm{oh},\mathrm{H}}})}$.

\emph{Multicast Throughput under Scheme
$\mathbb{M}_{\mathrm{o}\&\mathrm{h}}$:} According to
\emph{bottleneck principle}, we can obtain the final throughput
under the scheme $\mathbb{M}_{\mathrm{o}\&\mathrm{h}}$.
\end{proof}

\subsubsection{Scheme Using Both the P-AR  and Highway Systems, $\mathbb{M}_{\mathrm{p}\&\mathrm{h}}$}
By a similar procedure to the analysis of $\mathbb{M}_{\mathrm{o}\&\mathrm{h}}$, we can obtain
\begin{lem}\label{lem-throughput-parallel-highway}
Under the multicast scheme $\mathbb{M}_{{\mathrm{o}\&\mathrm{h}}}$, the multicast throughput is achieved as
$$\Lambda_{{\mathrm{p}\&\mathrm{h}}}(\lambda, n)=\min\left\{\frac{\mathbf{R}_{\mathrm{P-AR}}(\lambda,
n)}{\mathbf{L}(n_s,
\frac{1}{\mathbf{p}_{\mathrm{ph},\mathrm{P-AR}}})},
\frac{1}{\mathbf{L}(n_s,\frac{1}{\mathbf{p}_{\mathrm{ph},\mathrm{H}}})}\right\}.$$
\end{lem}

\end{document}